\documentclass[onecolumn]{IEEEtran}
\usepackage{cite}
\usepackage{amsmath,amssymb,amsfonts}
\usepackage{algorithmic}
\usepackage{graphicx}
\usepackage{textcomp}
\usepackage{xcolor}
\usepackage{dsfont}
\usepackage{xfrac}
\usepackage{mathtools}
\usepackage{mathabx}
\usepackage{stackengine}
\usepackage{tikz}
\usepackage{hyperref}
\usepackage{fontawesome}
\usetikzlibrary{arrows.meta, positioning, shapes.geometric, backgrounds, calc, fit}
\def\BibTeX{{\rm B\kern-.05em{\sc i\kern-.025em b}\kern-.08em
		T\kern-.1667em\lower.7ex\hbox{E}\kern-.125emX}}
\begin{document}

	\newtheorem{theorem}{Theorem}
\newtheorem{lemma}[theorem]{Lemma}
\newtheorem{corollary}[theorem]{Corollary}
\newtheorem{defn}{Definition}
\newtheorem{proposition}{Proposition}
\newtheorem{remark}{Remark}
\newtheorem{claim}{Claim}
\renewcommand{\tilde}{\widetilde}
\renewcommand{\hat}{\widehat}
\renewcommand{\bar}{\widebar}
\newcommand{\openone}{\mathds{1}}
\newcommand{\Mc}{\mathcal{M}}
\newcommand{\Ac}{\mathcal{A}}
\newcommand{\Gc}{\mathcal{G}}
\newcommand{\Vc}{\mathcal{V}}
\newcommand{\Qc}{\mathcal{Q}}
\newcommand{\Sbf}{\mathbf{S}}
\newcommand{\gcirc}{\mathring{g}}
\newcommand{\Gcirc}{\mathring{G}}
\newcommand{\Vcirc}{\mathring{V}}
\newcommand{\vcirc}{\mathring{v}}
\newcommand{\supp}{\text{supp}}
\newcommand{\Mcirc}{\mathring{M}}
\newcommand{\mcirc}{\mathring{m}}
\newcommand{\Shat}{\hat{S}}
\newcommand{\var}{\text{var}}
\newcommand{\cov}{\text{cov}}
\newcommand{\Scirc}{\mathring{S}}
\newcommand{\Ec}{\mathcal{E}}
\newcommand{\Xc}{\mathcal{X}}
\newcommand{\Yc}{\mathcal{Y}}
\newcommand{\Zc}{\mathcal{Z}}
\newcommand{\Tr}{\mathrm{Tr}}
\newcommand{\Pc}{\mathcal{P}}
\newcommand{\Fc}{\mathcal{F}}
\newcommand{\Cc}{\mathcal{C}}
\newcommand{\Dc}{\mathcal{D}}
\newcommand{\Hc}{\mathcal{H}}
\newcommand{\thetab}{\bar{\theta}}
\newcommand{\phat}{\hat{p}}
\newcommand{\mhat}{\hat{m}}
\newcommand{\MSE}{\mathrm{MSE}}
\newcommand{\that}{\hat{\theta}}
\newcommand{\polya}{P\'olya\ }
\newcommand\numberthis{\addtocounter{equation}{1}\tag{\theequation}}
\newcommand{\Pbar}{\bar{\mathbf{P}}}
\newcommand{\Phat}{\hat{P}}
\newcommand{\sigmab}{\bar{\sigma}}
\newcommand{\sigmat}{\tilde{\sigma}}
\newcommand{\Ptilde}{\tilde{P}}
\newcommand{\alphat}{\tilde{\alpha}}
\newcommand{\Qbar}{\bar{Q}}
\newcommand{\Wc}{\mathcal{W}}
\newcommand\barbelow[1]{\stackunder[1.4pt]{$#1$}{\rule{1ex}{.075ex}}}
\newcommand{\uepsilon}{\barbelow \epsilon}

	\title{Minimax bounds for watermarked and masked recursive discrete distribution estimation
		\thanks{The work in this paper was supported by SNSF Grant 10004381}
			\thanks{Email: \href{mailto:millen.kanabar@epfl.ch}{millen.kanabar@epfl.ch}, \href{michael.gastpar@epfl.ch}{michael.gastpar@epfl.ch}}
	}

	\author{%
		\IEEEauthorblockN{Millen Kanabar and Michael Gastpar}\\
		\IEEEauthorblockA{School of Computer and Communication Sciences,\\
			EPFL, Switzerland}
	}

	\maketitle
	
	\begin{abstract}
		Watermarking has been proposed as a way to identify synthetic samples in estimation settings where no metadata is available to distinguish them from real samples, but its precise effects remain unexplored. In the absence of a distinguishing mechanism, it has been shown that adding synthetic samples significantly reduces the marginal efficacy of new real samples. In this work, we study the minimax loss of such recursive discrete distribution estimation in the presence of watermarks in contrast to the unassisted and oracle-assisted losses. When the fraction of real samples vanishes asymptotically, we provide a lower bound that shows that it is impossible to improve performance by adding watermarks unless the false negative rate of detection also vanishes. Additionally, we show that in most regimes, the worst-case losses of a sequence of simple deterministic estimators match the corresponding lower bounds up to constants. Finally, we propose masking, a randomization procedure that narrows the gap in the remaining regimes to a Jensen gap. We conjecture that a tighter lower bound argument can close this gap.
	\end{abstract}

	\section{Introduction}
	In the generative machine learning literature, model collapse refers to the phenomenon where the underlying distributions of models trained only on their previous generation's outputs degrade with each iteration. This stands to reason, since finitely many samples disproportionately represent higher-probability outcomes. This leads to increasingly noisy estimates for lower-probability outcomes, until the output distribution eventually gets corrupted into a point mass by this sampling noise.
	
	Accumulating each batch of samples used for training and including fresh real samples in each batch can lead to eventual convergence to the ground truth, so long as real samples do not get `lost' in a sea of synthetic ones. Even so, if synthetic samples are not distinguished from real samples, the reinforcement of higher probability symbols in the empirical distribution of the corpus acts as noise and affects performance. In a previous work, we quantify the effect of this noise in the discrete distribution estimation setting \cite{kanabar_model_2026} via upper and lower minimax bounds. These bounds show that even in the regime where the expected number of real samples asymptotically grows to infinity, if synthetic sampling outpaces real sampling, the marginal utility of new real samples is greatly reduced.
	
	In the literature and in practice, watermarks have been proposed as a mechanism to identify synthetic samples and detect their provenances. When reliable, they can offer a useful way to identify and discard synthetic samples from the corpus. In this work, we use reductions from the watermarking-assisted setting to the unassisted setting studied in our previous work to find conditions where watermarking makes a difference. 
	
	The problem setting in our work is depicted in Figure \ref{fig:accumulate-workflow}. New samples added at stage $t$ are assumed to be drawn from the ground truth distribution $p$ with probability $\alpha_t$ and from the previous estimate $\Phat_{t-1}$ otherwise. We assume that the mixing factors $\alpha_t$ are known to the estimator beforehand. We also assume that samples do not change appreciably with the application of watermarks, and therefore ignore any distortion in the distribution of synthetic samples due to watermarking. This assumption mirrors watermarking in image models, where a small amount of structured noise is added as a watermark.
	
	The estimator is equipped with a watermark detector. Since watermarking is not a perfect procedure to identify sample provenance (see e.g.\ \cite{zhang_watermarks_2024, gowal_synthid-image_2025}), we model the outcome of the detector as a random variable correlated with the provenance of the sample with known output distributions given whether it is real or synthetic. This subsumes both detectors that provide a `confidence rating' instead of a binary outcome and training methods that might use additional metadata collected with each sample as soft watermarks.
	
	Through our lower bound, we observe that watermarking can improve performance significantly only if the false negative rate of the watermark detector shrinks with the fraction of real samples, matching the oracle-assisted setting (equivalently, perfect watermarking) up to constant factors whenever this rate is linear or faster. We also observe that a matching upper bound (up to constant factors) is achieved by simple deterministic estimators in most regimes. However, when synthetic sampling outpaces real sampling extremely rapidly, the performance of these estimators falls short of the lower bound. A similar effect was also observed in the unassisted setting in our previous work. With this in mind, we propose a randomization procedure (which we refer to as masking) that narrows it to a Jensen gap. We conjecture that a tighter lower bound might close this gap.

	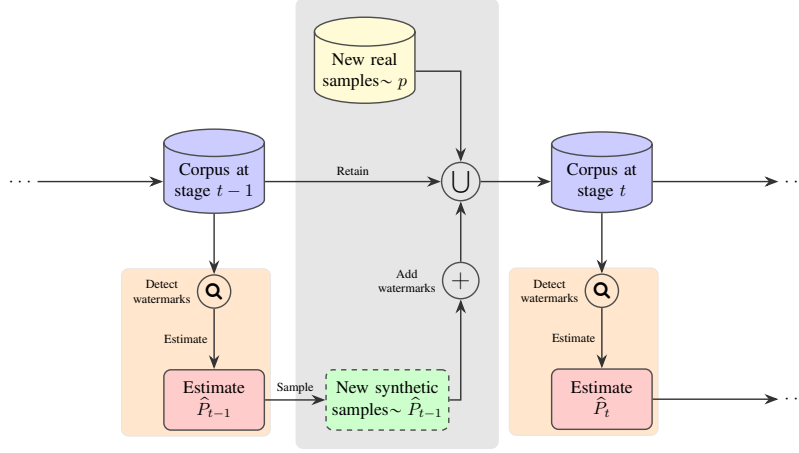
\begin{figure}
		\centering
		\resizebox{0.6\columnwidth}{!}{
			\begin{tikzpicture}[
				node distance=1.2cm and 1.5cm,
				data/.style={cylinder, draw=black!70, fill=blue!20, thick, aspect=0.3, minimum width=2.0cm, minimum height=1.4cm, align=center, shape border rotate=90},
				fresh/.style={cylinder, draw=black!70, fill=yellow!20, thick, aspect=0.3, minimum width=2.0cm, minimum height=1.4cm, align=center, shape border rotate=90},
				model/.style={rectangle, draw=black!70, fill=red!20, thick, minimum width=2.0cm, minimum height=1.2cm, align=center, rounded corners},
				gen/.style={rectangle, draw=black!70, fill=green!20, thick, dashed, minimum width=2.0cm, minimum height=1.2cm, align=center, rounded corners},
				arrow/.style={-{Stealth[scale=1.2]}, thick, draw=black!80},
				mix/.style={circle, draw=black!70, fill=gray!20, thick, inner sep=2pt},
				detect/.style={circle, draw=black!70, thick, inner sep=2pt}
				]
				
				\node[data] (Dn1) {Corpus at \\ stage $t-1$};
				\node[left=2.5cm of Dn1] (dots1) {$\cdot \cdot\cdot$};
				\node[detect, below=of Dn1] (Dw1) {\faSearch};
				\node[model, below= of Dw1] (Mn1) {Estimate \\ $\Phat_{t-1}$};
				\node[left=0 of Dw1, align=center, font=\scriptsize] (Dw1text) {Detect\\ watermarks};
				
				\node[gen, right= 1.2cm of Mn1] (Gn) {New synthetic  \\ samples$\sim \Phat_{t-1}$};
				\node[fresh, above right=0.7cm and 1.5cm of Dn1] (Fresh) {New real  \\ samples$\sim p$};
				
				\node[mix, right=3.5cm of Dn1, font=\large] (Mix) {$\bigcup$};

				\node[mix, below= of Mix] (Watermark) {$\bigplus$};
				\node[left= 0cm of Watermark, font=\scriptsize, align=center] {Add\\watermarks};
				
				\node[data, right=1.4cm of Mix] (Dn) {Corpus at\\ stage $t$};
				\node[detect, below=of Dn] (Dw) {\faSearch};
				\node[model, below=of Dw] (Mn) {Estimate \\ $\Phat_t$};
				\node[left=0 of Dw, align=center, font=\scriptsize] (Dwtext) {Detect\\ watermarks};
				\node[right=2.5cm of Dn] (dots2) {$\cdot \cdot\cdot$};
				\node[right=2.5cm of Mn] (dots3) {$\cdot \cdot\cdot$};
				
				\draw[arrow] (Dw1) -- node[left, font=\scriptsize] {Estimate} (Mn1);
				\draw[arrow] (Mn1) -- node[above, font=\scriptsize] {Sample} (Gn);
				
				\draw[arrow] (Dn1) -- (Dw1);
				\draw[arrow] (Dn1) -- node[above, font=\scriptsize] {Retain} (Mix);
				\draw[arrow] (Fresh) -| (Mix);
				\draw[arrow] (Gn) -| (Watermark);
				\draw[arrow] (Watermark) -- (Mix);
				\draw[arrow] (dots1) -- (Dn1);
				\draw[arrow] (Dn) -- (dots2);
				\draw[arrow] (Mn) -- (dots3);
				
				\draw[arrow] (Mix) -- (Dn);
				\draw[arrow] (Dw) -- node[left, font=\scriptsize] {Estimate} (Mn);
				\draw[arrow] (Dn) -- (Dw);
				
				\begin{scope}[on background layer]
					\node[fit=(Fresh)(Gn)(Mix), fill=gray!20, rounded corners, draw=gray!20, inner sep=10pt] (bg) {};
				\end{scope};
				
				\begin{scope}[on background layer]
					\node[fit=(Dw1)(Mn1)(Dw1text), fill=orange!20, rounded corners, draw=gray!20, inner sep=3pt] (bg2) {};
				\end{scope}
				
				\begin{scope}[on background layer]
					\node[fit=(Dw)(Mn)(Dwtext), fill=orange!20, rounded corners, draw=gray!20, inner sep=3pt] (bg3) {};
				\end{scope}
			\end{tikzpicture}
		}
		\caption{Stage $t$ in the accumulation workflow with watermarking}
		\label{fig:accumulate-workflow}
	\end{figure}
	\subsection{Notation}
	Deterministic quantities and functions are represented by lowercase Greek and Roman alphabets; random variables are represented by uppercase Roman alphabets. The probability simplex corresponding to the alphabet of size $k$ is denoted as $\Delta_k$. We use $[a:b]$ to denote $\{x\in \mathbb{N}: a\leq x\leq b\}$. The unit vector in with component $j$ equal to $1$ is denoted as $e_j$.
	
	The $j^{\mathrm{th}}$ component of a random or deterministic vector $P$ or $p$ is indexed by square brackets as $P[j]$ and $p[j]$ respectively, and the probability of an event $A$ under distribution $Q$ is denoted as $Q\{A\}$. The $n$-fold product of a distribution $Q$ is denoted as $Q^n$. 
	
	For $t\geq 0$ and functions $f, g$, we write $g(t) = o(f(t))$ if $\lim_{t\uparrow\infty}\sfrac{g(t)}{f(t)} = 0$; $g(t) = \mathcal O(f(t)) $ if there exists a constant $a\geq 0$ such that $g(t) \leq af(t) \ \forall t\geq 0$ and $g(t) = \Theta(f(t))$ if $g(t) = \mathcal{O} (f(t))$ and $f(t) = \mathcal{O} (g(t))$.
	
	\subsection{Related work}Model collapse has been widely observed empirically e.g.\ in \cite{shumailov_ai_2024, alemohammad2024selfconsuming, kazdan_collapse_2025, shi_closer_2026} and analyzed theoretically for discrete distributions \cite{seddik2024how, suresh2025rate, shumailov_ai_2024}, parametric models \cite{bertrand2024on, marchi_heat_2024, xu_probabilistic_2025, barzilai_when_2026}, regression and kernel density estimation \cite{fu_towards_2024, dohmatob_model_2024, gerstgrasser_is_2024, feng_beyond_2024, dohmatob2025strong, vu_what_2025, garg_preventing_2026}, and diffusion models \cite{fu_towards_2024, khelifa_error_2026}. Various mechanisms to mitigate model collapse have also been proposed, such as accumulating data \cite{gerstgrasser_is_2024, marchi_heat_2024, seddik2024how, fu_theoretical_2025, kazdan2025collapse}, watermarking \cite{shumailov_ai_2024, he_golden_2025, ji_overview_2026, racca_language_2026}, and data curation \cite{feng_beyond_2024, dohmatob_why_2025, amin2025escaping}.  Our masking procedure is inspired, in part, by recent work \cite{amin_learning_2026} that reduces recursive learning to the PU learning setting \cite{mansouri_learning_2025}. Additionally, in our setting, we only consider samples from the ground truth distribution and known estimates; similar bounds for settings such as \cite{jain_scaling_2024} that use samples from unknown distributions close to the ground truth remain open.
	
	\section{Problem Setup}
	Let $p\in \Delta_k$ be the distribution being estimated, and $\left(n_t, \alpha_t\right)_{t\geq 0}$ be a sequence over $\mathbb{N} \times [0, 1]$. The estimators might also be equipped with additional randomness $U_t$; the randomized estimates are available to the estimator at each subsequent stage of estimation. Sampling and estimation then follows the following procedure:
	\subsection{Sampling procedure}
	\paragraph{Initial samples and estimates}
	At stage $0$, the estimator receives $n_0$ samples $X_0^{n_0}$ sampled i.i.d.\ from $p$. 
	The estimator releases the estimate $\Phat_0 = \phat(X_0^{n_0}, U_0)$ and retains it for subsequent stages of estimation. 
	\paragraph{Recursive estimation}
	At every subsequent stage $t\geq 1$, random variables $Y_t^{n_t} \sim \left(\mathrm{Ber}(\alpha_t)\right)^{n_t}$ are drawn. These will determine the provenances of the samples. Subsequently, samples $X_t^{n_t}$ are drawn conditionally independently from $\bar{P}(X|\Phat_{t-1}, Y)$ given as
	\begin{align}
		\bar{P}(X = x|\Phat_{t-1}, Y) = \begin{cases}
			p[x] & Y = 1\\
			\Phat_{t-1}[x] & Y=0.
		\end{cases}
	\end{align}The estimate $\Phat_t := \phat_t\left((X_i^{n_i}, U_i)_{t\in[0:t]}\right)$ is released and retained for future estimation. For ease of notation, we define $Y_0^{n_0} $ as the all-ones vector. We will refer to the joint distribution of the whole sequence $\left(U_i, X_i^{n_i}\right)_{ i\in [0:t] }$ as $\Pbar_{t, p}$.

	\paragraph{Watermarking}Additionally, when a watermarking mechanism is present, we model the output of the watermarking detector as random variables $W_{t}^{n_t} \in \Wc^{n_t}$, with $|\Wc| = k_W$. Given $Y_t$, these are drawn conditionally independently from\begin{align}
		\bar{P}_{W_t|Y_t}(w|y) = \begin{cases}
			Q_{0, t}[w] & y=0\\
			Q_{1, t}[w] & y=1.
		\end{cases}
	\end{align} Once again, for ease of notation, we define $W_0^{n_0}$ as the all-ones vector. For conciseness, we will drop the dependence on $t$ and only write $Q_0$ and $Q_1$, however, the proofs do not depend on these staying constant.
	In this setting, the watermark-aided estimates $\Phat_t^{(\mathrm{w})}:=\phat_t^{(\mathrm{w})}\left((X_i^{n_i}, W_i^{n_i}, U_i)_{i\in[0:t]}\right)$ are released and retained for subsequent stages.
	We also add a superscript to the joint distribution of all the information available $\Pbar_{t, p}^{(\mathrm w)}$ to indicate the watermarking setting. The marginal distribution of the watermark $\alpha_t Q_1 +(1-\alpha_t)Q_0$ is denoted as $\bar{Q}_{\alpha_t}$. 
	
	For the rest of the paper, we use the term `watermark' to refer to the output of the detector.

	\paragraph{Oracle-assisted baseline} We consider estimators to be `oracle-assisted' when the sample provenances $(Y_i^{n_i})_{}$ are also available to the estimator via an oracle; this information is lost otherwise. In this setting, oracle-assisted estimates, denoted as $\Phat_t^{(\mathrm{o})} := \phat_t^{(\mathrm{o})}\left((X_i^{n_i}, Y_i^{n_i}, U_i)_{i\in[0:t]}\right)$, are released and retained in memory. As with watermarking, we also add the superscript $\Pbar_{t, p}^{(\mathrm o)}$ to indicate the oracle-assisted setting. 
	
	We use the unassisted and oracle-assisted performance of estimators as lower and upper baselines in this work: good watermarks and will get the loss as close to the oracle-assisted performance as possible, whereas a loss close to the unassisted performance will signal ineffective watermarks.
	
	\subsection{Measuring performance}
	We use the minimax expected $\ell_2^2$ loss as the figure of merit.
	\begin{defn}[Minimax expected loss] \label{def:minimax-loss}
		The minimax expected loss $r_{t, k}$ of the sequence $\phat_t$ is given as 
		\begin{align}
			r_{t, k} := \inf_{\phat_t} \sup_{p\in \Delta_k} E_{\Pbar_{t, p}} \left[\ell_2^2 \left(\Phat_t, p\right)\right].
		\end{align}
		The oracle-assisted and watermark-assisted minimax losses $r^{(\mathrm{o})}_{t, k}$ and $r^{(\mathrm{w})}_{t, k}$ are defined similarly.
	\end{defn}
	
	When estimating a discrete distribution $p \in \Delta_k$ from $n$ i.i.d.\ samples, the minimax $\ell_2^2$ loss is $\Theta(\sfrac{1}{n})$. Along similar lines, in the recursive estimation settings in this work, we informally refer to the multiplicative inverse of the expected $\ell_2^2$ losses (suitably normalized by constants) as the `effective sample size'.
	For example, using \cite[Theorems 1 \& 2]{kanabar_model_2026}, under most regimes in the unassisted setting, we find that the minimax loss is upper and lower bounded as
	\begin{align*}
		\frac{c_1}{\sum\limits_{i=0}^t n_i\alpha_{i}^2} \leq r_{t, k} \leq \frac{c_2}{\sum\limits_{i=0}^t n_i\alpha_{i}^2}
	\end{align*}
	for some positive constants $c_1$ and $c_2$ independent of the sequences $n_t, \alpha_t$. Consequently, we refer to the denominator $\sum\limits_{i=0}^t n_i\alpha_{i}^2$ as the effective sample size, and the individual terms $n_i\alpha_{i}^2$ as the effective sample size at stage $i$.

	\section{Main results}
	
	To help state the results, we first define the memory term that encapsulates the effect estimates have on lower bounds for future stages.
	
	\begin{defn}[Memory terms for the lower bound] \label{def:error-term-lb}
		The memory term $g_t(\rho)$ (and, similarly, $g_t^{(\mathrm{w})}(\rho)$) is defined as
		\begin{align}
			g_t(\rho) = \sup_{\substack{p\in \Delta_k\\x\in \Xc}}\Pbar_{t, p} \left\{\left|\Phat_{t-1}[x] - p[x]\right| \geq \rho\right\}.
		\end{align}The term $h_t(\rho)$ (and similarly $h_t^{(\mathrm{w})}(\rho)$) is defined as
		\begin{align}
			h_t(\rho) = \sum_{i=1}^{t}n_i\alpha_i g_i(\rho).
		\end{align}
	\end{defn}
	
	Notice that the memory terms depend on the deviation probabilities of the previous estimators. We comment on the necessity of these in our bounds in Section \ref{sec:masking-discussion} where we discuss the masking estimator.

	\subsection{Minimax bounds with watermarking}
	\begin{theorem}[Lower bound with watermarks] \label{thm:watermark-lb}
		Consider the watermark-aided recursive distribution estimation problem with tuples $(U_i, X_i^{n_i}, W_i^{n_i})_{i\in[0:t]}$ distributed according to $\Pbar^{(\mathrm{w})}_{t, p}$. There exists a constant $m$ independent of $t$ such that
		\begin{align}
			r_{t, k}^{(\mathrm{w})} \geq \frac{m}{\sum_{i=0}^{t}n_i\alpha_i^2\left(1+\chi^2\left(Q_1\middle\|\bar Q_{\alpha_{i}}\right)\right) + h_t^{(\mathrm{w})}(\sfrac{1}{4k})}
		\end{align} whenever $\sum_{i=0}^{t}n_i\alpha_i^2\left(1+\chi^2\left(Q_1\middle\|\bar Q_{\alpha_{i}}\right)\right) + h_t^{(\mathrm{w})}(\sfrac{1}{4k})\geq k/4$, with $h_t^{(\mathrm{w})}(\sfrac{1}{4k})$ defined in Definition \ref{def:error-term-lb}.
	\end{theorem}
	
	We also show that the lower bound at each stage $t$ is achievable by a sequence of deterministic estimators in most regimes under mild parametric assumptions. 
	\begin{theorem}[Upper bound with watermarks] \label{thm:watermark-ub}
		For the problem setup given above, there exists a sequence of deterministic estimators for which the worst-case loss is bounded above as
		\begin{align}
			\sup_p E[\ell_2^2(\Phat^{(\mathrm{w})}_t, p)] \leq \frac{3}{\sum_{i=0}^t n_i\alpha_i^2(1+\chi^2\left(Q_1\middle\|\bar Q_{\alpha_{i}}\right))}
		\end{align}
		whenever, for every $t\geq 1$,
		\begin{align}
			\frac{8\log\left(k_W\left(1+\chi^2\left(Q_1\middle\|\bar Q_{\alpha_{t}}\right)\right)\right)}{\min_w \bar Q_{\alpha_t}[w]} \leq n_t \leq \frac{1}{\alpha_t}\sum_{i=0}^t n_i\alpha_i^2.\label{eq:watermark-ub-conditions}
		\end{align}
		
	\end{theorem}
	
	We note that the hypothesis conditions \eqref{eq:watermark-ub-conditions} can be dispensed with in favor of more complicated results and proofs: the first inequality allows us to match a loose lower bound via concentration of the denominator, and the second is necessary only to ensure that the corresponding estimates lie in $\Delta_k$ with probability $1$ (as opposed lying outside $\Delta_k$ with only exponentially small probability). From the above bounds, we see that when the memory term is small, the effective sample size for batch $i$ in the watermarking setting is $n_i\alpha_{i}^2\left(1+\chi^2\left(Q_1\|\bar Q_{\alpha_i}\right)\right)$.
	
	\subsection{Performance of the masked estimator}\label{sec:masking-results}
	The above upper and lower bounds do not match when the memory term is not small. Making progress towards closing this gap, we show that probabilistically masking the output estimate reduces this to a Jensen gap.
	
	\subsubsection{The masking process}Let $B_t$ be a (not necessarily i.i.d.)\ process such that $B_{t} \sim \mathrm{Ber}\left(b_t\right)$ for some sequence $b_t$. Let $\Phat_{t}^{(\mathrm{pre})}$ be an unbiased pre-release estimate of $p$. Sample $S_t\sim \Phat_{t}^{(\mathrm{pre})}$. The released estimate is given as
	\begin{align*}
		\Phat_{t}[x] = \begin{cases}
			\Phat_t^{(\mathrm{pre})}[x] & B_t = 0\\
			\openone\{S_t = x\} & B_t = 1
		\end{cases}.
	\end{align*}
	We refer to this as masking to refer to the probabilistic replacement of the pre-release estimate with the $k$ deterministic distributions $e_x, x\in \Xc$. When the estimate is indeed masked at time $t$, each synthetic sample in batch $t+1$ is equal to $S_{t}$.
	
	\subsubsection{Performance}
	For conciseness of expression and proof, we state the upper bound in the unassisted setting, applying the masking process to a sequence of simple linear estimators. An equivalent expression with the corresponding $\chi^2$ divergence terms for each stage $i$ can be shown by masking the sequence of estimators proposed in Theorem $2$ instead. 
	\begin{theorem} \label{thm:masking estimator}
		There exists a sequence of estimators $\phat_t$ such that, for each $t\geq0$,
		\begin{align}
			&\sup_p E_{\Pbar_{t, p}}\left[\ell_2^2\left(\Phat_{t}, p\right)\right] \leq 2\sigma^{2}_t \cdot (1+c),  \label{eq:masking-ub}\\ 
			&\sigma_t^2 := E \left[\frac{1-\sfrac{1}{k}}{n_0 + \sum\limits_{i=1}^{t} n_i\alpha_i^2 \left(1-B_{i-1}\right) + n_i\alpha_iB_{i-1}}\right] \label{eq:masking-ub-def}
		\end{align} for any process $(B_i)_{i\geq 0}$ with marginal distributions $B_{t} \sim \mathrm{Ber}\left(b_t\right)$, as long as $b_i \leq c\sigma^2_i $ for some constant $c>0$ for every $i$.
		
	\end{theorem}
	
	The proofs of these theorems can be found in Appendix \ref{sec:proofs}
	
	\section{Discussion}
	
	To get a sense of the impact(s) of watermarking and masking, we compare our results to the previously known unassisted and oracle-assisted minimax bounds. Via \cite[Theorems 1 and 2 and Lemma 1]{kanabar_model_2026}, the effective sample sizes at stage $t$ are $n_t\alpha_t^2 +n_t\alpha_{t} g_t(\sfrac{1}{4k})$ for the lower bound and $n_t\alpha_t^2$ for the upper bound in the unassisted setting, and $n_t\alpha_t$ in both senses for the oracle-assisted setting. Watermarking therefore improves the effective samples size of stage $t$ by a factor of $1+\chi^2\left(Q_1\middle\|\Qbar_{\alpha_t}\right)$.

	\subsection{Watermarking in recursive settings}
	Since perfect watermarking is exactly equivalent to the availability of oracle information $(Y_i^{n_i})_{i\in [0:t]}$, we can compare effective sample sizes to find how fragile the benefit of watermarking is; a `good' watermark would lead to an increase in effective sample size by a factor of $\Omega(1/\alpha_{i})$ at each stage $i$. Given the statistics of the watermark, we have the following Proposition bounding the magnitude of this increment.
	\begin{proposition}
		For any $\alpha\in [0, 1]$, Let $\uepsilon = \min_w \sfrac{Q_0[w]}{Q_1[w]}$ Then $\chi^2\left(Q_1\middle\|\bar Q_{\alpha}\right)$ is bounded from above as
		\begin{align}
			1+\chi^2\left(Q_1\middle\|\bar Q_{\alpha}\right) \leq \frac{1}{\alpha + \uepsilon - \alpha\uepsilon} \leq \frac{1}{\max \{\alpha, \uepsilon\}}
		\end{align}
	\end{proposition}
	\begin{IEEEproof}
		The proof is direct. We have\begin{align*}
			1+\chi^2\left(Q_1\middle\|\bar Q_{\alpha}\right) =& \sum_{w\in\Wc}Q_1[w]  \frac{Q_1[w]}{\alpha Q_1[w] +(1-\alpha)Q_0[w]} \\
			=& \sum_{w\in\Wc} Q_1[w] \frac{1}{\alpha + (1-\alpha) \frac{Q_0[w]}{Q_1[w]}}\\
			\leq & \frac{1}{\alpha + (1-\alpha)\uepsilon}\cdot\sum_{w\in\Wc}Q_1[w].
		\end{align*}
	\end{IEEEproof}
	This bound provides a concise condition on $Q_0$ and $Q_1$ for watermarking to be useful: $\uepsilon_t$ should decay at least as quickly as $\alpha_{t}$. Since $Q_1[w] \leq 1$, this implies that the watermarking detector should `misidentify' synthetic samples as real samples (via at least one high confidence `real' signal $w^*$) with smaller and smaller probabilities as real samples become rarer and the estimates get closer to the ground truth distribution. 
	
	We present the exact performance in two well-known cases, where the watermarking detection statistics describe the output of the sample provenances $Y$ when passed through two well-known channels: the BEC and the BSC.
	
	\begin{claim}[BEC]
		Let the conditional distributions $Q_0$ and $Q_1$ denote the output distributions of a binary erasure channel (BEC) with error probability $\epsilon$ given inputs $0$ and $1$ respectively. $\chi^2(Q_1\|\bar{Q}_{\alpha})$ is then given as $\sfrac{1-\epsilon}{\alpha}$.
	\end{claim}
	
	This is exactly the behavior we expect, since the detection outputs for synthetic samples are either $0$ or `error'.  This might be difficult to implement in practice: watermarks still have relatively high false negative rates (see \cite{pang_no_2024, zhang_watermarks_2024, gowal_synthid-image_2025}). A more realistic watermarking scheme is described by a BSC:
	
	\begin{claim}[BSC]
		For $Q_0$ and $Q_1$ describing the conditional output distributions of a binary symmetric channel (BSC) with crossover probability $\epsilon<\sfrac{1}{2}$, the corresponding $\chi^2(Q_1\|\bar{Q}_{\alpha})$ is given as
		\begin{align*}
			(1-2\epsilon) \cdot\left(\frac{1}{\frac{\alpha}{1-\alpha} + \frac{\epsilon}{1-\epsilon}} - \frac{1}{\frac{\alpha}{1-\alpha}+ \frac{1-\epsilon}{\epsilon}}\right) \leq \frac{1}{\max\{\alpha, \epsilon\}}.
		\end{align*}
	\end{claim}
	Thus, whenever the probability of misidentification is not very small (here, when $\epsilon$ is not close to $\alpha$), watermarking only increases the effective sample size very marginally.

	\subsection{Masking estimators} \label{sec:masking-discussion}
	
	In a large selection of regimes, the term $h_t(\sfrac{1}{4k})$ (and similarly, $h_t^{(\mathrm{w})}(\sfrac{1}{4k})$) is dominated by the other term in the denominator \cite[Propositions 1 and 2]{kanabar_model_2026}. However, we see that there is a gap between the upper and lower bounds whenever the memory term $h_t(\sfrac{1}{4k})$ dominates the effective sample size. 
	
	\subsubsection{The linear estimation gap}
	Using Chebyshev's inequality, we can upper bound the memory term $g_i(\rho) \leq 16k^2\sigma^2_{i-1}$ and consequently, upper bound the total effective sample size by $n_0 + \sum_{i=1}^t16k^2n_i\alpha_{i}\sigma_{i-1}^2$.
	
	Now, any sequence of linear estimators of $p$ will lead to subgaussian estimates $\Phat_i^{(\mathrm{w})}$ (since the empirical distributions of each batch are subgaussian). Additionally, the bias of any such estimate has to be $\mathcal{O}(\var (\Phat_{i}^{(\mathrm{w})}))$, which is $o(1)$, i.e.\ the mean of these estimates is at an $o(1)$ distance from $p$. Consequently, the corresponding memory term $g_i$---the probability of deviating from $p$ by a constant distance $\sfrac{1}{4k}$---will be exponentially small due to the concentration of these subgaussians around their mean. If $\alpha_{i}$ (asymptotically) decreases quicker than $\sigma_{i-1}^2$, using linear estimators cannot lead to an effective sample size larger than $\sum_{i=0}^t n_i\alpha_{i}^2$, while our Chebyshev upper bound is of the order $n_0+ \sum_{i=1}^tn_i\alpha_{i}\sigma^2_{i-1}$.	
	We show that the non-linearity introduced by the masking procedure narrows this to a Jensen gap.
	
	\subsubsection{Partially bridging the gap via masking}Notice that when $b_i$ is indeed equal to $\sigma^2_i$ (or equivalently within constant factors), the expected value of the denominator in the RHS of \eqref{eq:masking-ub-def} is
	\begin{align}
		E\left[n_0 + \sum\limits_{i=1}^{t} n_i\alpha_i^2 \left(1-B_{i-1}\right) + n_i\alpha_iB_{i-1}\right]
		= \sum_{i=0}^tn_i\alpha_i^2 + \sum_{i=1}^t n_i\alpha_{i} \sigma_i^2.
	\end{align} While on the wrong side of Jensen's inequality, the loss of the sequence of masking estimators converges to the corresponding lower bound in \cite[Theorem 1]{kanabar_model_2026} in the regimes where this denominator concentrates around its expected value, thereby closing the gap. We conjecture that a more careful analysis which accounts for sample path averages over each of the events in the memory term can further close this gap by raising the lower bound for the regimes in which concentration does not occur.
	
	\subsubsection{Masking with distributions of support larger than $1$}
	For some alphabet $x$, let the released estimate $\Phat_{i-1}[x] = 0$. Conditioned on $x$ getting masked in this way, the marginal probability of drawing $x$ at stage $i$ is given as $\alpha_i p[x]$. The `corrected' empirical estimate $\frac{1}{\alpha_i}\phat_{\mathrm{emp}}[x]$ has variance exactly
	\begin{align}
		\frac{1}{n_i^2\alpha_{i}^2} \cdot n_i\cdot \alpha_{i}p[x](1-\alpha_{i}p[x]) = \frac{p[x](1-\alpha_{i}p[x])}{n_i\alpha_{i}},
	\end{align}which directly raises the effective sample size for the estimate of $p[x]$ from $n_i\alpha_i^2$ to $n_i\alpha_{i}$. Therefore, masking each estimate with distributions that have support size $1<\barbelow{k} < k$ would also lead to a similar upper bound, provided that masking is done often enough for each alphabet.
	We persist with masking using distributions of support size $1$ for simplicity of proof.
	
	\subsubsection{Masking as a watermark}
	Similar to implementations in text generation, masking itself acts as a watermark: given $\Phat_{t-1}[x] = 0$, all samples in stage $t$ with value $x$ can be directly classified as real with probability $1$. It does, however, distort the synthetic distribution, and therefore offers a more limited advantage compared to `invisible' watermarks. Practicality is further limited by the fact that our modeling assumption that only samples from $\Phat_{t-1}$ are mixed in at stage $t$ might no longer be valid in realistic settings.

	\bibliographystyle{ieeetr}
	\bibliography{ref_itw}

	\appendices
	
	\section{Proof of the Lower Bound} \label{sec:proofs}

	\subsection{The lower bound for watermarking}
	
	As in \cite{kanabar_model_2026}, we use the Paninski construction \cite{paninski_coincidence-based_2008} along with Assouad's method (see e.g.\ \cite{yu_assouad_1997}) to find lower bounds in the various settings.

	\begin{lemma}[Restated Lemma 7 from \cite{kanabar_model_2026}]\label{lem:lower-bound-partial}
		Let $Z \sim \Pbar_{t, p}$. Consider the set of vectors $\Vc = \{+1, -1\}^{k/2}$. With each vector $v\in\Vc$, associate the distribution
		\begin{align}
			p_v := p_u + \frac{\delta}{k}\begin{bmatrix}
				v\\-v
			\end{bmatrix} \label{eq:hard-distributions}
		\end{align} where $p_u$ is the uniform distribution. Let $\Pc_k:= \{p_v:v\in \Vc\}$.
		The $\ell_2^2$ minimax estimation risk of an estimator $\Phat_t  = \phat_t(Z)$  as described in Definition \ref{def:minimax-loss} is lower bounded as
		\begin{align}
			r_{t, k} \geq \frac{\delta^2}{4k}\left(1 {-} \sqrt{\frac{2}{k}\sum_{j=1}^{k/2} \max_{\substack{v\in\Vc:\\v_j=1}} D\Big(\Pbar_{t, p_v} \Big\| \Pbar_{t, p_{v-2e_j}}\Big)}\right) \label{eq:lower-bound-l2-kl-main}
		\end{align}
	\end{lemma}
	We are now ready to prove the lower bound.

	\begin{IEEEproof}[Proof of Theorem \ref{thm:watermark-lb}]
		The proof proceeds by showing an upper bound on $D\left(\Pbar_{t, p_v}^{(\mathrm{w})}\middle\|\Pbar_{t, p_{v-2e_j}}^{(\mathrm{w})}\right)$, with $p_v$ as defined in \eqref{eq:hard-distributions}. Define the conditional distributions (with some abuse of notation)
		\begin{align*}
			\tilde{\mathbf{P}}^{(\mathrm{w})}_{t, p}\left((x_t, w_t)^{n_t}|(x_i, w_i)^{n_i}_{i\in [0:t-1]}\right):=&\Pbar_{t, p}^{(\mathrm{w})}\left\{(X_t, W_t)^{n_t} = (x_t, w_t)^{n_t} \middle|(X_i, W_i)^{n_i}_{i\in [0:t-1]} = (x_i, w_i)^{n_i}_{i\in [0:t-1]}\right\},\\
			\tilde{P}^{(\mathrm{w})}_{t, p}\left(x_t, w_t|(x_i, w_i)^{n_i}_{i\in [0:t-1]}\right):=&\Pbar_{t, p}^{(\mathrm{w})}\left\{(X_t[j], W_t[j]) = (x_t, w_t) \middle|(X_i, W_i)^{n_i}_{i\in [0:t-1]} = (x_i, w_i)^{n_i}_{i\in [0:t-1]}\right\}& j\in [1:n_t],\\
			\text{and }\tilde{P}^{(\mathrm{w})}_{t, p, w_t}\left(x_t\middle|(x_i, w_i)^{n_i}_{i\in [0:t-1]}\right) :=& \Pbar_{t, p}^{(\mathrm{w})}\left\{X_t[j] = x_t \middle|(X_i, W_i)^{n_i}_{i\in [0:t-1]} = (x_i, w_i)^{n_i}_{i\in [0:t-1]}, W_t[j] = w_t\right\}& j\in [1:n_t].
		\end{align*} Then
		\begin{align*}
			D\left(\Pbar_{t, p_v}^{(\mathrm{w})}\middle\|\Pbar_{t, p_{v-2e_j}}^{(\mathrm{w})}\right) \overset{(a)}{=}& \sum_{i=0}^t D\left(\tilde{\mathbf{P}}_{i, p_v}^{(\mathrm{w})}\middle\|\tilde{\mathbf{P}}_{i, p_{v-2e_j}}^{(\mathrm{w})}\middle|\Pbar_{i-1, p_v}^{(\mathrm{w})}\right)\\
			=& \sum_{i=0}^t n_i E_{\Pbar_{i-1, p_v}^{(\mathrm{w})}}\left[D\left(\tilde{P}_{i, p_v}^{(\mathrm{w})}\middle\|\tilde{P}_{i, p_{v-2e_j}}^{(\mathrm{w})}\right)\right]\\
			=& \sum_{i=0}^t n_i E_{\Pbar_{i-1, p_v}^{(\mathrm{w})}}\left[E_{W\sim \bar{Q}_{\alpha_{t}}}\left[D\left(\tilde{P}_{i, p_v, W}^{(\mathrm{w})}\middle\|\tilde{P}_{i, p_{v-2e_j}, W}^{(\mathrm{w})}\right)\right]\right]\\
			\overset{(b)}{=}& \sum_{i=0}^t E_{W\sim \bar{Q}_{\alpha_{t}}}\left[n_i E_{\Pbar_{i-1, p_v}^{(\mathrm{w})}}\left[D\left(\tilde{P}_{i, p_v, W}^{(\mathrm{w})}\middle\|\tilde{P}_{i, p_{v-2e_j}, W}^{(\mathrm{w})}\right)\right]\right] \numberthis\label{eq:watermark-kl-ub}
		\end{align*} where $(a)$ follows from the chain rule of KL Divergences and $(b)$ follows because of the independence of the distribution of the watermark and the samples themselves. Observe that
		\begin{align}
		\tilde{P}^{(\mathrm{w})}_{t, p, w} =\tilde \alpha_{i, w} p + (1-\tilde{\alpha}_{i, w})\Phat_{t-1}^{(\mathrm{w})} \label{eq:watermarking-mixing}
		\end{align}where \begin{align}
		\tilde{\alpha}^{(\mathrm{w})}_{i, w} := & \frac{\alpha_{i} Q_1[w]}{\bar Q _{\alpha_{i}}[w]} & i\in[0:t].
		\end{align}Thus, conditioned on each watermark detection output $W_t[j] = w_t$, the distributions of the samples $X_t[j]$ are equal to the marginal distributions of unwatermarked samples with equivalent mixing factors $\tilde{\alpha}^{(\mathrm{w})}_{i, w_t}$. We can therefore directly use results from the unwatermarked setting to bound \eqref{eq:watermark-kl-ub}.
		From \cite[Equation 17]{kanabar_model_2026}, for each $W=w$,
		\begin{align}
			n_i E_{\Pbar_{i-1, p_v}^{(\mathrm{w})}}\left[D\left(\tilde{P}_{i, p_v, W}^{(\mathrm{w})}\middle\|\tilde{P}_{i, p_{v-2e_j}, W}^{(\mathrm{w})}\right)\right] \leq \frac{8\delta^2}{k\left(\frac{9}{16} - \delta^2\right)}\left( n_i{\left(\tilde{\alpha}_{i, w}^{(\mathrm{w})}\right)^2} + n_i\tilde{\alpha}^{(\mathrm{w})}_{i, w} g_i (\sfrac{1}{4k})\right).
		\end{align} Consequently,
		\begin{align*}
			D\left(\Pbar_{t, p_v}^{(\mathrm{w})}\middle\|\Pbar_{t, p_{v-2e_j}}^{(\mathrm{w})}\right) \leq \frac{8\delta^2}{k\left(\frac{9}{16} - \delta^2\right)}E_{W\sim \bar{Q}_{\alpha_{t}}}\left[\sum_{i=0}^t n_i{\left(\tilde{\alpha}_{i, W}^{(\mathrm{w})}\right)^2} + n_i\tilde{\alpha}^{(\mathrm{w})}_{i, W} g_i (\sfrac{1}{4k}) \right].
		\end{align*} Since
		\begin{align}
			E_{W\sim \bar{Q}_{\alpha_{t}}}\left[\left(\tilde{\alpha}^{(\mathrm{w})}_{i, W}\right)^2\right] =&  E_{W\sim\Qbar_{\alpha_t}}\left[\alpha_t^2\left(\frac{Q_1[w]}{\Qbar_{\alpha_t}[x]}\right)^2\right] = \alpha_{i}^2\left(1+\chi^2\left(Q_1\middle\|\bar Q_{\alpha_{i}}\right)\right) \\
			\text{and } E_{W\sim \bar{Q}_{\alpha_{t}}}\left[\tilde{\alpha}^{(\mathrm{w})}_{i, W}\right] =& E_{W\sim \Qbar_{\alpha_t}}\left[\alpha_{t}\frac{Q_1[w]}{\Qbar_{\alpha_t}[x]}\right] =  \alpha_{i},
		\end{align} choosing \begin{align}
			\delta^2 = \frac{k}{64\sum_{i=0}^t \left(n_i\alpha_i ^2\left(1+\chi^2\left(Q_1\middle\|\bar Q_{\alpha_{i}}\right)\right)  + n_i\alpha_{i}g_i(\sfrac{1}{4k})\right)}\label{eq:lb-delta-def}
		\end{align}gives us
		\begin{align*}
			D\left(\Pbar_{t, p_v}^{(\mathrm{w})}\middle\|\Pbar_{t, p_{v-2e_j}}^{(\mathrm{w})}\right) \leq \frac{1}{4}
		\end{align*}whenever $\sum_{i=0}^{t}n_i\alpha_i^2\left(1+\chi^2\left(Q_1\middle\|\bar Q_{\alpha_{i}}\right)\right) + h_t^{(\mathrm{w})}(\sfrac{1}{4k})\geq k/4$. Invoking Lemma \ref{lem:lower-bound-partial} then leads to a lower bound of 
		\begin{align}
			r^{(\mathrm{w})}_{t, k} \geq \frac{1}{512\sum_{i=0}^t n_i\alpha_i ^2\left(1+\chi^2\left(Q_1\middle\|\bar Q_{\alpha_{i}}\right)\right)  + h_t(\sfrac{1}{4k})},
		\end{align} completing the proof.
	\end{IEEEproof}
	
	\section{Proofs of the Upper Bounds}
	
	\subsection{Useful Lemmas}
	
	We also have the following Lemma showing the performance of a linear estimator in the presence of recursive sampling.
	
	\begin{lemma}[Restated Lemma 9 from \cite{kanabar_model_2026}]\label{lem:upper-bound-uncombined}
		Let $\Phat_{0}$ be an unbiased estimate of $p$ with a variance $\eta^2[x]\cdot p[x](1-p[x])$ for each component $x\in\Xc$. Let samples $Z^{n} \sim (\alpha p + (1-\alpha)\Phat_0)^n$. Denote the empirical distribution of samples $X^n$ as $\phat_{\mathrm{emp}}(X^n)$. Then the intermediate estimate
		\begin{align*}
			\Phat_{1, \mathrm{int}} =\phat_{1, \mathrm{int}}(X_1^{n_1}) := \frac{1}{\alpha}\left(\phat_{\mathrm{emp}}(X_1^{n_1}) - (1-\alpha)\Phat_0\right) \numberthis\label{eq:conditionally-unbiased-estimator}
		\end{align*}has expectation $p$ and is component-wise uncorrelated with $\Phat_0$.
		Additionally, each component $x\in \Xc$ has variance\begin{align*}
			E\left[\left(\Phat_{1, \mathrm{int}}[x] {-} p[x]\right)^2\right] = \frac{p[x](1{-}p[x])}{n_1\alpha^2}\left(1 {-} (1{-}\alpha)^2\eta^2[x]\right). \numberthis\label{eq:self-consuming-upper-bound}
		\end{align*}
	\end{lemma}
	
	Finally, we state the following elementary lemma that will be useful in combining estimates when upper bounds on their variances are known.
	\begin{lemma} \label{lem:linear-combinations}
		Let $\hat{\theta}_1$ and $\hat{\theta}_2$ be two unbiased and uncorrelated estimates of a parameter $\theta \in \Theta$, and let their variances be bounded above by $\sigma_1^2$ and $\sigma_2^2$. We can then construct a linear combination $\thetab$ of $\hat{\theta}_1$ and $\hat{\theta}_2$ that is unbiased and with variance $\sigmab^2$ bounded above as $\sigmab^2 \leq \sfrac{1}{\left(\frac{1}{\sigma_1^2} + \frac{1}{\sigma_2^2}\right)}.$
	\end{lemma}
	
	\subsection{Proof of the watermarking upper bound}
	
	\begin{IEEEproof}[Proof of Theorem \ref{thm:watermark-ub}]
		The proof proceeds inductively. Assume that at stage $t-1$, each component of $\Phat_{t-1}$ is unbiased and has variance upper bounded as
		\begin{align*}
			\var \left(\Phat^{(\mathrm{w})}_{t-1}[x]\right) \leq \sigmab_{t-1}^2[x] = \frac{3 \cdot p[x](1-p[x])}{\sum_{i=0}^t n_i\alpha_i^2(1+\chi^2\left(Q_1\middle\|\bar Q_{\alpha_{i}}\right))}.
		\end{align*}It is easy to verify that the empirical estimator satisfies these conditions for $t=0$. 
		
		Let the number of samples in batch $t$ watermarked as $w$ be $N_{t}[w]$, and let the corresponding samples be $X_{t, w}^{N_{t}[w]}$. Recall from \eqref{eq:watermarking-mixing} that given $\Phat_{t-1}^{(\mathrm{w})}$ and $W_t = w$, $X_{t, w}^{N_{t}[w]}$ is a conditionally i.i.d.\ vector with marginal distributions $\alphat^{(\mathrm{w})}_{t, w} p + (1-\alphat_{t, w}^{(\mathrm{w})})\Phat_{t-1}^{(\mathrm{w})}$. Therefore, invoking Lemma \ref{lem:upper-bound-uncombined}, for each $x\in \Xc$, the conditional variance of the $x^{\mathrm{th}}$ component of intermediate estimates $\Phat_{t, w, \mathrm{int}} := \phat_{t, \mathrm{int}}\left(X_{t, w}^{N_{t}[w]}\right)$ given $W_t^{n_t}$ and $\Phat^{(\mathrm{w})}_{t-1}$ is bounded above by
		\begin{align*}
			\var\left(\Phat_{t, w, \mathrm{int}}[x]\middle|W_t^{n_t}, \Phat_{t-1}^{(\mathrm{w})}\right)
			\leq \frac{p[x](1-p[x])}{N_{t}[w]\left(\alphat_{i, w}^{(\mathrm{w})}\right)^2}.
		\end{align*}For each $x\in \Xc$, all estimates $\Phat_{t, w, \mathrm{int}}[x], w\in\Wc$ are also unbiased and conditionally independent given $W_t^{n_t}$ and $\Phat^{(\mathrm{w})}_{t-1}$. The combined intermediate estimate $\Phat_{t, \mathrm{int}} = \sum_{w\in \Wc}\frac{N_t[w] \left(\alphat_{i, w}^{(\mathrm{w})}\right)^2}{\sum_{w'\in \Wc}N_t[w']\left(\alphat_{i, w'}^{(\mathrm{w})}\right)^2} \cdot \Phat_{t, w, \mathrm{int}}$ therefore has a conditional variance bounded above as
		\begin{align}
			\var\left(\Phat_{t, \mathrm{int}}[x]\middle|W_t^{n_t}, \Phat_{t-1}^{(\mathrm{w})}\right) \leq \frac{p[x](1-p[x])}{\sum_{w\in\Wc}N_{t}[w]\left(\alphat_{i, w}^{(\mathrm{w})}\right)^2}, \label{eq:watermark-ub-unconcentrated}
		\end{align} (via Lemma \ref{lem:linear-combinations}) where the denominator has an expected value $n_t\alpha_{t}^2\left(1+\chi^2\left(Q_1\middle\|\bar Q_{\alpha_{t}}\right)\right)$. Now,
		\begin{align*}
		\Pbar_{t, p}^{(\mathrm{w})}\left\{\sum_{w\in\Wc}N_{t}[w]\left(\alphat_{i, w}^{(\mathrm{w})}\right)^2 \leq \frac{1}{2} n_t\alpha_{t}^2\left(1+\chi^2\left(Q_1\middle\|\bar Q_{\alpha_{t}}\right)\right)\right\} \overset{(a)}{\leq}& \Pbar_{t, p}^{(\mathrm{w})}\left\{\bigcup\limits_{w\in\Wc} \left\{N_t[w] \leq \frac{n_t\Qbar_{\alpha_t}[w]}{2}\right\} \right\}\\
		\overset{(b)}{\leq}& \sum_{w\in \Wc} \Pbar_{t, p}^{(\mathrm{w})}\left\{N_t[w] \leq \frac{n_t\Qbar_{\alpha_t}[w]}{2}\right\}\\
		\leq& k_W \max_{w\in \Wc} \Pbar_{t, p}^{(\mathrm{w})}\left\{N_t[w] \leq \frac{n_t\Qbar_{\alpha_t}[w]}{2}\right\}\\
		\overset{(c)}{\leq}& k_W\exp\left(-n_t\frac{\min_z\Qbar_{\alpha_t}[w]}{8}\right) \\
		\overset{(d)}{\leq}& \frac{1}{1+\chi^2(Q_1\|\Qbar_{\alpha_t})}, \numberthis \label{eq:watermarks-concentration}
		\end{align*}
		where $(a)$ follows from $W_t[i] \sim \Qbar_{\alpha_t}$ and $\alphat_{i, w}^{(\mathrm{w})} = Q_1[w]/\Qbar_{\alpha_t}[w]$, $(b)$ follows from the union bound, $(c)$ follows from Hoeffding's inequality, and $(d)$ follows from the lower bound assumption on $n_t$ in \eqref{eq:watermark-ub-conditions}. Define the pre-release estimate as
		\begin{align*}
			\Phat_{t, \mathrm{pre}} := \begin{cases}
				\Phat_{t, \mathrm {int}} & \text{if } \sum_{w\in\Wc}N_{t}[w]\left(\alphat_{i, w}^{(\mathrm{w})}\right)^2 > \frac{1}{2} n_t\alpha_{t}^2\left(1+\chi^2\left(Q_1\middle\|\bar Q_{\alpha_{t}}\right)\right)\\
				\phat_{t, \mathrm{int}}(X_t^{n_t}) & \mathrm{otherwise}.
			\end{cases}
		\end{align*} Consequently, from \eqref{eq:watermark-ub-unconcentrated}, Lemma \ref{lem:upper-bound-uncombined}, and \eqref{eq:watermarks-concentration}, \begin{align*}
			E\left[\var\left(\Phat_{t, \mathrm{pre}}[x]\middle|W_t^{n_t}, \Phat_{t-1}^{(\mathrm{w})}\right)\right] \leq& \frac{2\cdot p[x](1-p[x])}{n_t\alpha_{t}^2\left(1+\chi^2\left(Q_1\middle\|\bar Q_{\alpha_{t}}\right)\right)}\\
			& + \frac{p[x](1-p[x])}{n_t\alpha_t^2}\cdot \Pbar_{t, p}^{(\mathrm{w})}\left\{\sum_{w\in\Wc}N_{t}[w]\left(\alphat_{i, w}^{(\mathrm{w})}\right)^2 \leq \frac{1}{2} n_t\alpha_{t}^2\left(1+\chi^2\left(Q_1\middle\|\bar Q_{\alpha_{t}}\right)\right)\right\}\\
			\leq& \frac{3 \cdot p[x](1-p[x])}{\sum_{i=0}^t n_i\alpha_i^2(1+\chi^2\left(Q_1\middle\|\bar Q_{\alpha_{i}}\right))}.
		\end{align*}$\Phat_{t, \mathrm{pre}}$ is unbiased and component-wise uncorrelated with $\Phat_{t-1}^{(\mathrm{w})}$ (since $\phat_{t, \mathrm{int}}(X_t^{n_t})$ and $\Phat_{t, \mathrm{int}}$ are also unbiased and component-wise uncorrelated). Using Lemma \ref{lem:linear-combinations} again to combine $\Phat_{t, \mathrm{pre}}$ and $\Phat_{t-1}^{(\mathrm{w})}$ completes the induction step.
	\end{IEEEproof}

	\subsection{Proof for the masking estimator upper bound}
	Recall that the masking procedure takes an unbiased pre-release estimate $\Phat_t^{(\mathrm{pre})}$ and masking decision $B_t \in \{0, 1\}$ as inputs and releases the estimate
	\begin{align*}
		\Phat_{t}[x] = \begin{cases}
			\Phat_t^{(\mathrm{pre})}[x] & B_t = 0\\
			\openone\{S_t = x\} & B_t = 1
		\end{cases}
	\end{align*}where $S_t \sim \Phat_t^{(\mathrm{pre})}$.
	\begin{IEEEproof}[Proof of Theorem \ref{thm:masking estimator}]
		Fix $B_i {=} u_i {\in} \{0{,} 1\}$ for each $i{\in} [0{:}t]$. Mirroring the notation in Lemma \ref{lem:upper-bound-uncombined}, define $\eta_{i}^2[x]$ as the scalar such that  \begin{align*}
			\var\left(\Phat_{i}[x]\middle|(B_i)_{i\in [1:t]} = (u_i)_{i\in [1:t]}\right)= \eta_{i}^2[x] \cdot p[x](1-p[x]). 
		\end{align*} If $\Phat_{t}^{(\mathrm{pre})}$ is unbiased, $E[\openone\{S_t = x\}] = p$ and therefore $\Phat_{t}$ is also unbiased. 
		
		Using the empirical estimator $\phat_{0, \mathrm{int}} = \phat_{\mathrm{emp}}$ in stage $0$, and subsequently $\phat_{i, \mathrm{int}}$ from Lemma \ref{lem:upper-bound-uncombined} in stage $i\geq1$, we have
		\begin{align*}
			\var\left(\Phat_{0, \mathrm{int}}\middle|(B_i)_{i\in[0:t]} = (u_i)_{i\in [0:t]}\right)\leq & \frac{p[x](1-p[x])}{n_0} & i=0\\
			\var \left(\Phat_{i, \mathrm{int}}[x]\middle|(B_j)_{j\in[0:t]} = (u_j)_{j\in [0:t]}\right) \leq& \frac{p[x](1-p[x])}{n_i\alpha_i^2}\cdot(1-(1-\alpha_i)^2\eta_{i-1}^2[x]) & i\geq 1
		\end{align*}
		These are all unbiased and uncorrelated with the previous estimate. Therefore, using Lemma \ref{lem:linear-combinations}, we can find the pre-release estimate $\Phat_t^{(\mathrm{pre})}$ with a variance upper bound of
		\begin{align*}
			\var \left(\Phat_t^{(\mathrm{pre})}\middle|(B_i)_{i\in[0:t]} = (u_i)_{i\in [0:t]}\right) \leq \frac{p[x](1-p[x])}{n_0 + \sum\limits_{i=1}^t\frac{n_i\alpha_i^2}{1-(1-\alpha_{i}^2)\eta^2_{i-1}[x]}}
		\end{align*}
		
		We now consider two cases: 
		\begin{enumerate}
			\item The variance of $\Phat_i[x]$ given $B_i = 1$ is
			\begin{align*}
				E\left[(1-p[x])^2\Phat_i^{(\mathrm{pre})}[x] + p[x]^2(1-\Phat_{i}^{(\mathrm{pre})})\right] = (1-p[x])^2p[x] + p[x]^2(1-p[x]) = p[x](1-p[x]).
			\end{align*}Thus, if $u_{i-1} = 1$, then $\eta_{i-1}^2[x] = 1$ for all $x\in \Xc$, and consequently
			\begin{align*}
				1{-}(1{-}\alpha_{i})^2\eta_{i-1}^2[x] {=} 2\alpha_{i}.
			\end{align*}
			\item If $B_i=0$, we simply use the trivial upper bound of $1$ for $1-(1-\alpha_{i})^2\eta^2_{i-1}[x]$.
		\end{enumerate}
		Consequently, given $(B_i)_{i\in[0:t]} = (u_i)_{i\in [0:t]}$, we have
		\begin{align*}
			\var \left(\Phat_t^{(\mathrm{pre})}\middle|(B_i)_{i\in[0:t]} = (u_i)_{i\in [0:t]}\right) \leq \sigmab_t^2 {:=}  \frac{p[x](1-p[x])}{n_0 {+} \sum\limits_{i=1}^t \frac{n_i\alpha_{i}}{2}u_i {+} n_i\alpha_i^2 (1{-}u_i)}
		\end{align*}
		
		Taking the sum over $x$ and expectation over the joint distribution of $B_i$ and noticing that masking can only increase the loss by $b_t$, and therefore only a factor of $c$ when $b_t \leq c\sigmab_t^2$ concludes the proof.
	\end{IEEEproof}

\end{document}